\pdfoutput=1
\PassOptionsToPackage{noend}{algpseudocode}
\documentclass[11pt]{article}

\usepackage[T1]{fontenc}
\usepackage[utf8]{inputenc}
\usepackage[a4paper,margin=25mm]{geometry}
\usepackage[protrusion=true,expansion=false]{microtype}

\usepackage{amsmath,amssymb,amsthm}
\usepackage{mathtools}
\usepackage{cases}
\usepackage{bbold}
\usepackage{dsfont}

\usepackage{graphicx}
\usepackage{xcolor}
\usepackage{xspace}
\usepackage{algorithm}
\usepackage{algpseudocode}
\usepackage{cite}
\usepackage[hidelinks]{hyperref}

\theoremstyle{plain}
\newtheorem{theorem}{Theorem}
\newtheorem{lemma}[theorem]{Lemma}
\theoremstyle{definition}
\newtheorem{problem}{Problem}

\graphicspath{{./}{./}}

\newcommand{\NetEC}{\mbox{NetEC}\xspace}
\newcommand{\lcaf}{\mathsf{lca}}

\newcommand{\arrow}{\rightarrow}

\DeclareMathOperator{\True}{\mathsf{True}\xspace}
\DeclareMathOperator{\False}{\mathsf{False}\xspace}
\DeclareMathOperator{\Unknown}{\mathsf{Unknown}\xspace}
\DeclareMathOperator{\Lop}{\mathsf{L}}
\DeclareMathOperator{\Mop}{\mathsf{M}}
\DeclareMathOperator{\troot}{\mathsf{root}\xspace}
\DeclareMathOperator{\Map}{\mathsf{M}\xspace}

\newcommand{\Mapf}{F}
\DeclareMathOperator{\dset}{\mathsf{Dup}\xspace}
\DeclareMathOperator{\dsetupper}{\mathsf{UDup}\xspace}
\DeclareMathOperator{\epi}{\mathsf{Epi}\xspace}
\DeclareMathOperator{\upperepi}{\mathsf{UEpi}\xspace}
\DeclareMathOperator{\EC}{\mathsf{EC}\xspace}

\newcommand{\opusUniNet}{2023/51/B/ST6/02792}

\title{Episode Clustering in Phylogenetic Networks}

\author{
Pawe{\l} G\'{o}recki\thanks{Corresponding author: \href{mailto:gorecki@mimuw.edu.pl}{gorecki@mimuw.edu.pl}}\quad
Agnieszka Mykowiecka\quad
Jarosław Paszek\\[0.6em]
\small Faculty of Mathematics, Informatics, and Mechanics, University of Warsaw,\\
\small Banacha 2, 02-097 Warsaw, Poland
}

\date{\small Preprint version submitted to ECCB 2026 before peer review.\\
A revised version has been accepted for publication in \emph{Bioinformatics} as part of the ECCB 2026 Proceedings.}

\begin{document}
\maketitle

\begin{abstract}
The classical duplication episode clustering (EC) model introduced by Guigó et al. in the 1990s provides a foundational approach for inferring genomic duplication events crucial to understanding genome evolution.
This model clusters single gene duplications from a collection of gene trees at locations in the species tree to minimize the total number of such locations, called duplication episodes.
Here, we introduce \NetEC, a novel extension of this problem to phylogenetic networks. To solve \NetEC, we first develop a polynomial-time dynamic programming (DP) algorithm for testing whether a given set of network nodes can serve as episode locations. We then propose a main inference algorithm that utilizes this DP component to optimize the episode count; while the feasibility test runs in polynomial time, the full optimization has exponential worst-case complexity, and an optional heuristic mode is provided for larger instances. We also propose an extended episode analysis procedure that identifies additional genomic duplication candidates below reticulation nodes, complementing the main algorithm by resolving potential upward clustering of duplications induced by reticulation. We evaluate our method on simulated data and on an empirical Pandanales dataset comprising over 29,000 gene trees, demonstrating exact and accurate inference of genomic duplication events even in the presence of multiple reticulations.
\end{abstract}


\medskip

\section{Introduction}

Phylogenetic networks have emerged as a robust framework for representing complex evolutionary relationships~\cite{Huson2010} that traditional tree-based models cannot adequately capture. Unlike phylogenetic trees, networks accommodate reticulate events such as hybridization, horizontal gene transfer, and recombination, creating multiple pathways of inheritance.

Whole-genome duplications (WGDs) represent particularly significant evolutionary events that have shaped eukaryotic genomes~\cite{saceWGDgeneretention, ohno}. In phylogenetic networks, WGDs introduce additional complexity: hybridization following independent WGDs in parental lineages generates intricate patterns of gene family evolution, as exemplified by polyploid plant and fungal lineages~\cite{pmid27479829domestication, wolfe1997molecular}.

The concept of duplication episode clustering, developed initially for tree-based phylogenies~\cite{Guigo1996}, aims to identify genomic locations where multiple gene duplications co-occurred, suggesting large-scale duplication events like WGDs. Informally, an \emph{episode} is a node in the species phylogeny to which one or more gene duplications are assigned; episode clustering seeks the smallest set of such nodes explaining all observed duplications across a collection of gene trees. While episode clustering for trees has been extensively studied with efficient polynomial-time algorithms~\cite{burleigh2008locating, Luo2011, RME}, its extension to phylogenetic networks remains largely unexplored. The tree-based formulations, including interval models~\cite{RME}, unrooted variants~\cite{Paszek2018}, and path-constrained clustering~\cite{RMP:pmid31425045}, all rely on the fundamental assumption of a unique evolutionary path between any two nodes, an assumption violated in networks.

Extending episode clustering to networks introduces several theoretical and computational challenges. First, the presence of reticulation nodes means that gene duplications can be assigned to episodes along multiple alternative evolutionary histories, creating an exponential space of possible episode configurations. Second, the biological interpretation of episodes in networks requires careful consideration: should episodes be defined on the network structure itself, or on the individual tree-like scenarios (display trees) embedded within the network? Third, the computational complexity of network reconciliation suggests that episode clustering in networks may require fundamentally different algorithmic strategies than those used in tree-based approaches.

In our recent work~\cite{Gorecki24}, we introduced a polynomial-time dynamic programming algorithm using three-valued logic for episode feasibility testing in the tree-based setting. Here we extend this approach to phylogenetic networks, where multiple alternative evolutionary histories through reticulation nodes complicate episode assignment.

In this work, we formalize \NetEC (Network Episode Clustering), the problem of identifying duplication episodes in phylogenetic networks, and introduce the first algorithmic solution for it in the presence of reticulate evolution. Building on~\cite{Gorecki24}, we extend the three-valued logic dynamic programming approach to handle multiple evolutionary scenarios in networks. While the previous approach resolves uncertainty about
  gene-species leaf assignments, here we show how to resolve uncertainty about which evolutionary paths through reticulation nodes
  are compatible with observed duplication patterns. Our approach leverages unfolded networks, a transformation
  expanding a network into a tree-like structure representing all scenario choices, and formulates the problem using
  scenario functions that determine reticulation path choices.

The key contributions of this work are: (1) a formal definition of episode clustering for phylogenetic networks based on scenario-based gene-network reconciliation; (2) a polynomial-time algorithm for testing episode feasibility in networks via dynamic programming; (3) an exact algorithm for \NetEC with practical heuristics for identifying optimal solutions; and (4) experimental validation demonstrating the method's effectiveness in detecting WGDs in simulated and real biological networks with complex reticulate histories.

\section{Basic Definitions}

We collect here the basic terminology used throughout the paper.

\newcommand{\taxaset}{\Theta}

A (phylogenetic) \emph{network on a set of taxa $\taxaset$} is a directed acyclic graph $N = (V(N), E(N))$ such that
(1) there is a unique node, called a \emph{root}, such that there is a directed path from the root to any node in $N$ and (2) leaves of $N$, i.e., nodes of indegree 1 and outdegree 0, are bijectively labelled by the elements from $\taxaset$.
The leaf labelling is a function $\Lambda \colon L(N) \rightarrow \taxaset$, where $L(N)$ is the set of all leaves in $N$.
A node of $N$ is a reticulation if it has an indegree of at least 2. Nodes that are not leaves are \emph{internal}; internal nodes that are not reticulations are called \emph{tree nodes}.
By $R(N)$ we denote the set of all reticulations in $N$. If $(s, t) \in E(N)$ then $s$ is called a \emph{parent} of $t$ and $t$ is called a \emph{child} of $s$.
A network is binary if its leaves, root, and the remaining nodes have degrees 1, 2 and 3, respectively.
$N$ is \emph{semi-binary} if, in addition, it may contain semi-binary nodes with indegree at most 1 and outdegree 1, including the case where the root has exactly one child. A semi-binary node $v$ of indegree $1$ can be contracted by: (1) removing $v$ and the edges incident with $v$, and (2) inserting a new edge connecting the parent of $v$ with the child of $v$.
If $v$ has indegree $0$, then after removing $v$, the child of $v$ becomes the new root.
We say that an edge $(s,t) \in E(N)$ is a \emph{reticulation edge} if $t \in R(N)$.
If a node $s$ has exactly one child, then the child is denoted by $s'$ and if $s$ has exactly two children, then the children are denoted by $s'$ and $s''$. In the latter case, we say that $s'$ is a sibling of $s''$ and vice versa.
If there is a directed path from $s$ to $t$ (following edge directions from root toward leaves), then we say that $t$ is \emph{visible} from $s$, denoted as~$s \succeq t$.
A network on $\taxaset$ is \emph{tree-child} if every non-leaf node has a non-reticulation child.

A \emph{gene tree over a set of taxa} $\taxaset$ is defined similarly to the network but with two differences: it has no reticulation nodes, and the leaf labeling $\Lambda$ is not required to be a bijection.

\subsection{Unfolded Network}\label{sec:unfolded}

The unfolded network is a specific multi-labeled tree (MUL-tree, \cite{Huber2006}) obtained by unfolding reticulation nodes. Here, we briefly recall the unfolding construction from~\cite{wawerka22}.

For a network $N$, and, for each $i = 0,1,\dots,|R(N)|$ we define a pair $(N_i, \pi_i)$ as follows. Let $N_0 = N$ and $\pi_0$ be the identity function on $V(N)$.
Then, $(N_{i+1}, \pi_{i+1})$ is obtained from $(N_i, \pi_i)$ by the unfolding operation:
(1) pick a reticulation $p \in R(N_i)$ such that no other re\-ti\-culation is visible from $p$. Let $N_i^p$ be the subtree of $N_i$ rooted at $p$ and let $u$ denote an arbitrary parent of $p$; (2) copy $N_i^p$, (3) remove the edge $(u, p)$ and add a new edge from $u$ to the root of the copy of $N_i^p$. For $s \in N_{i+1}$ let $\pi_{i+1}(s) := \pi_i(s)$ if $s \in N_i$ and $\pi_{i+1}(s) := \pi_i(t)$, if $s$ is a copy of a node $t \in N_{i}$. See the middle example in Fig.~\ref{fig:intro}.

$N_{|R(N)|}$ is called the \emph{unfolded network} of $N$ and denoted $\hat{N}$. Also, by $\pi$ we denote~$\pi_{|R(N)|}$. In other words, $\pi$ is the projection of unfolded network nodes to the source nodes in $N$.

It follows from~\cite{wawerka22} that the unfolded network $\hat{N}$ of $N$ is a semi-binary tree. There is a one-to-one correspondence between root-leaf paths in $N$ and root-leaf paths in $\hat{N}$ established by $\pi$, i.e.,
if $P=p_1,p_2,\dots,p_m$ is a root-leaf path in $\hat{N}$, then $\pi(P)=\pi(p_1),\pi(p_2),\dots,\pi(p_m)$ is the corresponding root-leaf path~in~$N$.

Note that the size of $\hat{N}$ grows exponentially with the size of $R(N)$. It will become evident later that our algorithms do not use unfolded networks directly.

For a gene tree $G$ over a taxa set $Y$ from a network $N$, a \emph{scenario} is a function $\xi \colon L(G) \rightarrow L(\hat{N})$ such that for each $l \in L(G)$, $\Lambda(l) = \Lambda(\xi(l))$. An \emph{lca-mapping in a scenario} $\xi$ is a function $\Map_\xi \colon V(G) \rightarrow V(\hat{N})$ that extends $\xi$:
  \begin{numcases}{\Map_\xi(g) =}
      \xi(g) & if $g \in L(G)$, \label{eq:lca_leaf}\\
      \lcaf(\Map_\xi(g'), \Map_\xi(g'')) & otherwise, \label{eq:lca_internal}
  \end{numcases}
 where $\lcaf(x,y)$ is the least common ancestor of nodes $x$ and $y$.

A scenario can be visualized as an embedding of a gene tree into an unfolded network.~See~Fig.~\ref{fig:intro}.

\begin{figure}[t]
\centering
\includegraphics[width=\columnwidth]{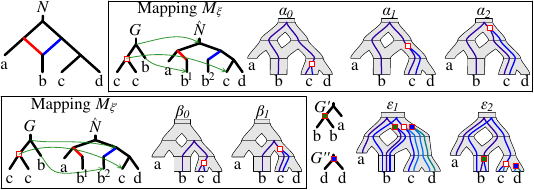}
\caption{\emph{Top-left:} A network $N$ with one reticulation. \emph{Framed Panels:} A gene tree $G$, an unfolded network $\hat{N}$ obtained from $N$, and two possible scenarios $\xi$ and $\xi'$. A scenario $\xi$ routes gene tree leaves through one parent edge of the reticulation node, while $\xi'$ routes them through the other, yielding two distinct histories. There are two lca-mappings $M_\xi$ and $M_{\xi'}$ and five valid mappings $\alpha_0-\alpha_2$, and $\beta_0-\beta_1$ shown here as embeddings with duplications (squares). Embeddings $\alpha_0$ and $\beta_0$ correspond to $M_\xi$ and $M_{\xi'}$. \emph{Bottom-right:} two alternative episode clusterings $\epsilon_1$ and $\epsilon_2$ arising from gene trees $G$, $G'$ and $G''$. In $\epsilon_1$, one episode at the root's right child groups three duplications at a single network node, whereas $\epsilon_2$ assigns each duplication to a separate episode at leaves $b$, $c$ and $d$ as singletons. Our objective is to minimize the number of episode locations, so $\epsilon_1$ is preferred over $\epsilon_2$. Note that in embeddings, duplications are placed on the edges directly above their mapped nodes, reflecting that these events precede the corresponding speciation (or leaf) in evolutionary time.}
\label{fig:intro}
\end{figure}

An internal node $g \in V(G)$ is a \emph{duplication} in a scenario $\xi \colon L(G) \rightarrow L(\hat{N})$, or \emph{$\xi$-duplication}, if $\Map_\xi(g) = \Map_\xi(a)$ for a child $a$ of $g$.
The remaining internal nodes, we call \emph{$\xi$-speciations}.

\subsection{Episode Clustering Problems}

We now formalize the model of duplication episodes for phylogenetic networks. Unlike the tree-based episode clustering model, our network-based approach must account for the multiplicity of evolutionary scenarios induced by reticulation nodes. The model admits all biologically plausible evolutionary scenarios with minimal duplication events.

In phylogenetic networks, episode clustering requires careful handling of the interplay between scenario selection (choosing paths through reticulations) and duplication assignment (determining episode locations). The model of gene duplication episodes allows relocating a gene duplication from its lca-mapping node to one of its ancestors, subject to additional constraints that preserve both biological soundness and scenario consistency. For a gene tree $G$ over a set of taxa from a network $N$, a mapping $\Mapf \colon V(G) \rightarrow V(\hat{N})$ is \emph{valid}, if there is a scenario $\xi$ between $G$ and $N$ such that
\begin{itemize}
\item $\Mapf(a) \preceq \Mapf(b)$ if $a \preceq b$ (time consistency),
\item $\Mapf(a)=\pi(\Map_{\xi}(a))$ for any $\xi$-speciation or a leaf $a$ (fixed speciations),
\item $\Mapf(a) \succeq \Map_{\xi}(a)$ for any $\xi$-duplication node $a$
  (duplication can be raised),
\item and, $\Mapf(a) \prec \Map_{\xi}(b)$ for any speciation node $b$ such that
  $a \prec b$ (fixed number of $\xi$-duplications).
\end{itemize}

Note that if $\Mapf$ is valid on $\xi$, then
$\Mapf|_{L(G)} = \xi$. Therefore, if a valid $\Mapf$ on $\xi$ is given, we often use the term $\Mapf$-duplication (or $\Mapf$-speciation) instead of $\xi$-duplication ($\xi$-speciation, respectively).
Additionally, $\Mapf$-duplications $g$ satisfying $\Mapf_\xi(g) = \Map_\xi(g)$ are called \emph{lca-duplications} (when the context $\Mapf$ is clear).

In practice, we are interested in locations of events in the original network $N$, rather than $\hat{N}$.
Therefore, let $\Mapf^* \colon V(G) \to V(N)$ be the mapping $\Mapf^*(g)=\pi(\Mapf(g))$, i.e., the mapping that allows a direct embedding of a gene tree into the network.
We also write that a node $g \in V(G)$ is \emph{assigned} to $s \in V(N)$ if $\Mapf^*(g)=s$.
By $\dset(\Mapf) \subseteq V(G)$, we denote the set of all $\Mapf$-duplications, while by $\epi(\Mapf) \subseteq V(N)$ we denote the set of all $\Mapf$-duplication locations in the network $N$, defined as \[ \epi(\Mapf) = \{ \Mapf^*(g) \colon g \in \dset(\Mapf)\}. \]

Fig.~\ref{fig:intro} provides an example of valid mappings that define an evolutionary scenario that can be represented as a tree with additional decoration of nodes. For details on the formal modeling of evolutionary scenarios, refer to~\cite{Gorecki2006dls}.

Assume that $\Mapf_i \colon G_i \to \hat{N}$ is a valid mapping between a gene tree $G_i$ over a set of taxa from a network $N$, for every $i \in \{1,2,\dots,n\}$.
Every element in $\bigcup_i \epi(\Mapf_i)$ denotes the location of
a multiple gene duplication event in $N$. We refer to these nodes as \emph{duplication episodes} or simply \emph{episodes}. An episode represents a location in the network where multiple gene duplications have occurred, potentially across different scenarios corresponding to different gene families.

\begin{problem}[Network Episode Clustering, \NetEC]\label{NetEC}
Given a phylogenetic network $N$ and a collection of rooted gene trees $G_1, G_2,\dots, G_k$ over the set of taxa present in $N$.
Compute the minimum number of duplication episodes,
denoted by $\EC(G_1, G_2, \dots, G_k, N)$, in the set of all
valid mappings $\Mapf_1, \Mapf_2, \dots, \Mapf_k$ between $G_i$ and $\hat{N}$, respectively.
\end{problem}
This problem is solvable in linear time when $N$ is a tree~\cite{Paszek2018}; its complexity for networks remains open.

\section{Methods}

Our approach to \NetEC proceeds in three stages. First, we develop a constrained feasibility test that determines whether a given set of network nodes can serve as episode locations for a single gene tree. Second, we extend this test to compute the minimum episode clustering for a single gene tree. Third, we generalize to multiple gene trees. The key insight enabling our approach is the observation that while a phylogenetic network may encode exponentially many scenarios, we can efficiently reason about episode feasibility by working directly on the network structure.

\subsection{Net-Episode Feasibility}
\label{sec:EFPDP}

We start with a fundamental constrained problem. Given a gene tree and a phylogenetic network, we ask whether there exists a scenario and a corresponding valid mapping such that the set of duplication episodes is contained in a given fixed set of candidate episode locations in the network.

\begin{problem}[Net-Episode Feasibility]
\label{problem:feasibility}
Given a gene tree $G$ over a set of taxa from a network $N$ and $X \subseteq V(N)$.
Does there exist a valid mapping $\Mapf$ from $G$ to $\hat{N}$
such that $\epi(\Mapf) \subseteq X$?
\end{problem}
If a gene tree $G$ satisfies the above property, we call $G$ \emph{$X$-feasible} with respect to $N$. If the context is clear, we omit the reference to $X$.
Reticulation nodes can be excluded from episode locations: any duplication assigned to a reticulation is indistinguishable from one assigned to its unique child, so it can always be reassigned downward without changing the episode count. We therefore assume w.l.o.g.\ that $X \subseteq V(N) \setminus R(N)$, which reduces the search space of our algorithms.

\newcommand{\unf}{\hat{N}}
A challenge in network reconciliation is that duplications may be constrained to episodes not in the candidate set $X$. Still, these constraints may only become apparent when traversing upward in the network toward the root. To handle this uncertainty, we employ Łukasiewicz's Three-Valued Logic $\text{Ł}_3$~\cite{lukasiewicz1970selected}, which extends Boolean logic with an $\Unknown$ value representing situations where episode assignment is deferred to higher levels of the network.
This three-valued logic has constants $\True$, $\False$, and $\Unknown$ ordered linearly as $\False < \Unknown < \True$, with binary operators $\vee$ (disjunction, $\max$) and $\wedge$ (conjunction, $\min$). The unary operators are defined as: $\Lop$ (certainty operator), where $\Lop(x) = \True$ if $x = \True$ and $\Lop(x) = \False$ otherwise; and $\Mop$ (possibility operator), where $\Mop(x) = \False$ if $x = \False$ and $\Mop(x) = \True$ otherwise.

\newcommand{\Ft}{$\mathds F$}
\newcommand{\Ut}{$\mathds U$}
\newcommand{\Tt}{$\mathds T$}

\newcommand{\Ftn}{\mathds F}
\newcommand{\Utn}{\mathds U}
\newcommand{\Ttn}{\mathds T}

For a node $v$ of a tree $T$, by $T|v$ we denote the subtree of $T$ rooted~at~$v$.
Recall that $s'$ and $s''$ denote the children of $s$, and similarly for $g$.
To simplify the notation, we assume that the set $X \subseteq V(N)$ is fixed.

The dynamic programming formulas to solve Net-Episode Feasibility are depicted in Fig.~\ref{fig:dp}.
\begin{figure}[!thbp]
\small
\textbf{For $g \in V(G)$ and $s \in V(N)\setminus R(N)$:}
\begin{numcases}{\delta(g,s)=}
    \delta^*(g,s)  & \text{$g$ internal, $s \in X$,} \label{v1d} \\
    \delta^*(g,s) \wedge \Unknown & \text{$g$ internal, $s \notin X$,} \label{v2d} \\
     \False & \text{otherwise,} \label{v3d}
\end{numcases}
\begin{equation}
\delta^*(g,s)=\epsilon(g',s) \wedge \delta^\downarrow(g'',s) \vee \epsilon(g'',s) \wedge \delta^\downarrow(g',s),\label{vdstar}
\end{equation}
\begin{numcases}{\delta^\downarrow(g,s)=\epsilon(g,s) \vee}
    \bigvee_{c \in ch(s)} \Mop \delta^\downarrow(g,c) & \text{$s$ non-leaf, $s \in X$,} \label{v2dd} \\
    \bigvee_{c \in ch(s)} \delta^\downarrow(g,c)
     & \text{otherwise,} \label{v1dd}
\end{numcases}
\begin{numcases}{\sigma(g,s)=}
    \Lop \big ( \delta^\downarrow(g',s') \wedge \delta^\downarrow(g'',s'') \nonumber\\
    \vee  \delta^\downarrow(g',s'') \wedge \delta^\downarrow(g'',s') \big)
         & \text{$g$ internal, $s$ tree-node,} \label{v1s} \\
     \True & \text{$g$ leaf labelled $s$,} \label{v2s} \\
     \False & \text{otherwise,} \label{v3s}
\end{numcases}
\textbf{For $s\in R(N)$:} $f(g,s)=f(g,s')$ for $f \in \{\delta, \delta^*, \delta^\downarrow, \sigma\}$.\refstepcounter{equation}\label{vdretcase}\hfill(\theequation)

\textbf{Auxiliary function:} $\epsilon(g,s)=\sigma(g,s) \vee \delta(g,s).$
\caption{Dynamic programming (DP) formulas for Net-Episode Feasibility. In $\delta^\downarrow$, $ch(s)$ is the set of children of a node $s$, the empty set if $s$ is a leaf. The auxiliary function $\epsilon(g,s)$ captures whether $g$ can be placed at $s$ either as a speciation or a duplication.}
\label{fig:dp}
\end{figure}

For a gene tree $G$ over $N$ and a node $g \in V(G)$,
let $\Mapf \colon V(G|g) \to V(\hat{N})$ be a valid mapping. We say that $\Mapf$ is \emph{feasible} for $(g,s,X)$ if and only if $\Mapf^*(g) \preceq s$  and $\epi(\Mapf) \subseteq X$.
Feasible mappings represent episode scenarios that correspond to partial solutions to the instance of Net-Episode Feasibility that have all duplications present in $X$.
We say that an $\Mapf$-duplication $d$ in a gene tree $T$ is \emph{upper} if
the path from $d$ to $g$ consists of $\Mapf$-duplications.
The set of all upper $\Mapf$-duplications we denote $\dsetupper(\Mapf) \subseteq \dset(\Mapf)$ and their locations in $N$ we denote by $\upperepi(\Mapf) \subseteq \epi(\Mapf)$.
We write that $\Mapf \colon V(G|g) \to V(\hat{N})$ is \emph{weakly feasible} for $(g,s,X)$ if and only if $\Mapf^*(g) \preceq s$, $g \in \dsetupper(\Mapf)$ and $\epi(\Mapf) \setminus \upperepi(\Mapf) \subseteq X$.

In weakly feasible mappings we constrain only non-upper duplications present in $G|g$, while the upper duplications will be elements of episode $s' \in X$, such that $s \prec s'$, if such $s'$ exists. This situation is modeled by $\Unknown$ value returned from $\delta^\downarrow(g,s)$ and $\delta(g,s)$~calls.

Informally, the meaning of DP formulas can be understood as follows:
$\delta(g, s)$ is $\True$ if there is a feasible mapping $\Mapf$ for $(g,s,X)$
such that $s \in X$, $g$ is an $\Mapf$-duplication assigned to $s \in X$, and all duplications are assigned to the episodes from $X$. Similarly, $\sigma(g, s)$ is $\True$ if there is a feasible mapping $\Mapf$, where $\Mapf^*(g)=s$ and $g$ is a speciation or a leaf. Next, $\delta(g, s)$ is $\Unknown$ if there is no feasible mapping
for $(g,s,X)$, however, there is a weakly feasible mapping $\Mapf$ for $(g,s,X)$,
where $g$ is an $\Mapf$-duplication assigned to $s \notin X$, and all non-upper duplications from $G|g$ are assigned to the episodes from $X$.
Note that $\sigma(g, s)$ cannot be $\Unknown$ since speciation nodes are fixed. Moving on, $\delta^\downarrow(g, s)$ is $\True$ if there is a feasible mapping for $(g,s,X)$, where all duplications are assigned to the episodes from $X$. Lastly, $\delta^\downarrow(g, s)$ is $\Unknown$ if the condition for $\delta^\downarrow(g, s) = \True$ is not met. However, there is a weakly feasible mapping for $(g,s,X)$.

To solve Net-Episode Feasibility, we apply $\delta^\downarrow$ on the roots.
\begin{theorem}[Correctness]
\label{lem:DPcormain}
Given a gene tree $G$ over a network $N$ and $X \subseteq V(N)\setminus R(N)$.
$G$ is $X$-feasible if and only if $\delta^\downarrow(\troot(G),\troot(N))$ is $\True$.
\end{theorem}

Finally, the time and space complexity of solving Episode Feasibility by the DP algorithm is $O(|V(G)| \cdot |V(N)|)$.

For feasible instances, the backtracking can identify the subset of nodes from $X$ that contribute to the optimal solution by quantifying the number of duplications, i.e., episode sizes for each episode. These duplication counts provide insights into the significance of each inferred episode.

See Appendix for the proofs.

\subsection{Solution for a single gene tree and the general case}

\newcommand{\FE}{\text{FE}}
\newcommand{\BestEC}{\text{BestEC}}
First, we describe the main algorithm to solve \NetEC for instances with a single gene tree.

Alg.~\ref{alg:NetECsingle} extends~\cite{Gorecki24} by first identifying \emph{fixed episodes} (nodes present in every solution) through testing feasibility with each node excluded; if removal makes the instance infeasible, that node is fixed. The main loop then searches for a set $C \subseteq V(N)\setminus\Phi$ of size $b-|\Phi|-1$ such that $G$ is $(C\cup\Phi)$-feasible, updates $b$ via DP backtracking, and terminates when no such $C$ exists.

\begin{algorithm}[!htbp]
\caption{Solution to \NetEC with a single gene tree}\label{alg:NetECsingle}
\small
\begin{algorithmic}[1]
\Require A gene tree $G$ over a network $N$
\Ensure $\EC(G,N)$
\State $\Phi \gets \emptyset$ \Comment{Init: the set of fixed episodes $\Phi$}
\For{every node $v$ in $V(N)\setminus R(N)$} \Comment{Identify fixed episodes}
    \If{there is no feasible mapping for $V(N) \setminus \{v\}$}
        \State Add $v$ to $\Phi$ ($v$ is a fixed episode)
    \EndIf
\EndFor
\State $b \gets |V(N)\setminus R(N)|$ \Comment{The initial maximal number of episodes}
\While{$\mbox{there is } C \subseteq V(N) \setminus \Phi \mbox{ of the size } b-|\Phi|-1$ and
      $G$ is $(C \cup \Phi)$-feasible} \Comment{The main loop}
    \State $b \gets$ the $\EC$ cost by backtracked DP from Fig.~\ref{fig:dp}.
\EndWhile
\State Optional backtracking: compute the episode sizes by counting duplications at given episode $s \in X $ when~\eqref{v1d} is reached in DP.\label{line:backtracking}
\State \Return $b$
\end{algorithmic}
\end{algorithm}

The correctness of the algorithm follows from the fact that if there is no set $X$ of size $b-1$ such that $G$ is $X$-feasible, then there is no set of any size smaller than $b$ that satisfies the property. Since $b$ represents the number of episodes from some valid mapping, it is also minimal. Therefore, when the algorithm terminates, $b=\EC(G,N)$.

The algorithm's worst-case time complexity is $n^2m + \sum_{k=f}^{n-f}{\binom{n-f}{k}nm} = O(nm \cdot 2^n)$, where $f$ is the size of the set of fixed episodes ($f=|\Phi|$), $n$ denotes the number of nodes in $N$, and $m$ denotes the number of nodes in $G$.
Despite the exponential time complexity, in our experiments on both simulated and empirical data, we were able to compute exact solutions after only a few executions of the main loop.

To identify the optimal solution within the main loop, enumerating all possible combinations of size $b-f-1$ from the set of episode candidates $V(N) \setminus \Phi$ may be time-consuming for larger instances. To address this issue, we propose a heuristic approach that randomly samples combinations of size $b-f-1$ if $\binom{n-f}{b-f-1}$ is large, similarly to our previous solution from~\cite{Gorecki24}. In our experiments, the heuristic mode was not reached.

To solve \NetEC in a general case, we transform the problem to a single gene tree case.
Given a collection of gene trees $G_1, G_2, \dots, G_k$ and a network $N$. Let $\omega$ be a new species, called \emph{outgroup}, not present in $N$.
We first add the outgroup to every input tree.
Let $G_1^\omega=(G_1,\omega)$ and $G^\omega_i=((G_i,\omega),G^\omega_{i-1})$, for $i>1$.
Let $N^\omega$ be a network obtained from $N$ by inserting a new root and a leaf labelled by the outgroup and connecting the new root with the leaf and the root of $N$.
Then, by $\omega$-\NetEC we define the problem \NetEC with a single gene tree.
\begin{lemma}\label{lem:genNetEC}
Given at least two gene trees $G_1, G_2, \dots, G_k$ and a network $N$ such that $\omega \notin L(N)$. Then, $X \subseteq V(N)$ is the set of episodes that yields the solution of \NetEC for $G_1, G_2, \dots, G_k$ and $N$ if and only if
$X \cup \{\troot(N^\omega)\}$ is the set of episodes that yields the solution to the instance $G^\omega_k$ and $N^\omega$ of $\omega$-\NetEC.
\end{lemma}
For $k$ gene trees, the construction merges them into a single tree of size $O(k+M)$,
where $M = \sum_{i=1}^{k}|V(G_i)|$, so the overall time complexity becomes $O\!\left(n (k+M) 2^n\right)$ where $n = |V(N)|$.

\subsection{Post-evaluation: extended episodes analysis}

Our experiments revealed that duplications located below a reticulation node may cluster at the \emph{lowest stable ancestor} (LSA) of the reticulation (i.e., the lowest common ancestor of all parents) or higher nodes. This artifact results from the flexible mapping model, which permits scenarios utilizing both reticulation edges, combined with the episode minimization objective. The left part of Fig.~\ref{fig:netecepi1000} illustrates this: duplications from WGD events at $B$ are moved to the LSA (marked by a star).  Alg.~\ref{alg:extepisodes} addresses this by iteratively testing non-episode nodes in post-order and adding those whose episode size exceeds a threshold (e.g., the average episode size). Processing nodes bottom-up ensures that the lowest high-signal candidates are identified first, preventing their duplications from being absorbed by higher-level episodes.

\begin{algorithm}[!htbp]
\caption{Extended episodes analysis}\label{alg:extepisodes}
\small
\begin{algorithmic}[1]
\Require A set of gene trees over a network $N$.
\Ensure The set of episodes with extended episodes
\State Merge gene trees into a single tree (see Lemma~\ref{lem:genNetEC}).
\State Compute the set of episodes $Q$ using Alg.~\ref{alg:NetECsingle}.
\For{$s \in V(N) \setminus Q$ in post-order}
\State{Compute the episode size of $s$ using DP with $X = Q \cup \{s\}$}
\If{the size of episode $s$ exceeds threshold}
\State Add $s$ to $Q$ \Comment{$s$ becomes an extended episode}
\EndIf
\EndFor
\State \textbf{return} $Q$
\end{algorithmic}
\end{algorithm}

\section{Results}

We evaluate \NetEC on simulated data with known WGD ground truth and on an empirical Pandanales dataset comprising over 29,000 gene trees on a network with two reticulations. All experiments were performed using the \NetEC tool (\url{https://github.com/ppgorecki/netec}), a publicly available software package implementing the algorithms presented in this study; the input data, scripts, and parameter settings used to reproduce the experiments are also available via the same repository. The tool operates in two modes: \emph{discovery} mode for direct inference of duplication episodes, and \emph{verification} mode where users specify candidate WGD locations for validation.

\subsection*{Simulated data}

\emph{Data preparation.}
We sampled ultrametric tree-child networks of the height $1.8 \times 10^9$ years based on species trees with $20$ leaves and one reticulation event using the procedure described in~\cite{Rutecka24}, and selected a network $N$ where the reticulation connects temporally proximate lineages. The network $N$ is depicted in~Fig.~\ref{fig:netecepi1000} and~\ref{fig:netecepi1000fixb}, where the reticulation node representing hybridization event is marked as black circle.

We identified three candidate locations for WGD events, denoted $A$, $B$, and $C$.  Location $A$ is positioned close to but independent of the hybridization event, representing WGD that occurs without direct relationship to the hybridization process. Location $B$ is positioned below the reticulation in $N$ and therefore appears twice in $\hat{N}$, once on each descendant branch following hybridization. Location $C$ is evolutionarily distant, positioned on a branch unrelated to the hybridization event.

We analyzed seven scenarios in total: no WGD ($\emptyset$), single WGDs at locations $A$, $B$, or $C$, and all pairwise combinations ($AB$, $AC$, $BC$). Whole-genome duplication events were simulated following~\cite{gdilp}. For a given node $v$ in tree $\hat{N}$, a WGD event at $v$ was modeled by replacing the corresponding subtree $\hat{N}|v$ with its duplicated copy $(\hat{N}|v, \hat{N}|v)$. When the duplication event occurred below a reticulation, this substitution was applied symmetrically to both unfolded parts of the network. Consequently, we generated seven species trees from the unfolded network $\hat{N}$ for each WGD  scenario.

Our simulation study consists of several phases to infer~gene trees under biologically realistic conditions. A total of 1000 replicates were generated for each of seven species trees,~yielding 7,000 gene trees in total.
True gene trees were simulated using SimPhy~\cite{simphy} under a multilocus
coalescent~model incorporating incomplete lineage sorting and gene duplication/loss. We
employed a duplication/loss rate of $2 \times 10^{-10}$ events per year and an
effective population size of $N_e = 10^7$. Gene tree heights were drawn from a
lognormal distribution ($\mu = 1.5$, $\sigma = 1$). DNA sequences of 1000 bp were simulated along the true gene trees using AliSim~\cite{iqtree2} under the GTR+$\Gamma$ model. Model parameters were sampled from empirical Dirichlet priors following~\cite{molloy2020}. Insertions and deletions were simulated with rates of 0.03 and 0.09 per substitution, respectively, with lengths following a Zipfian distribution (exponent 1.7, maximum length 50). Sequences were aligned using MAFFT~\cite{mafft}, and maximum likelihood trees were estimated from these inferred alignments using PhyML~\cite{phyml}
under the GTR+$\Gamma$ model with estimated para\-meters.
The unrooted trees were rooted by~midpoint-plateau rooting implemented in URec~\cite{urec}. All datasets were processed by the \NetEC tool within 2 hours~on~a~standard workstation.

\emph{Results.}
The results of the discovery mode (Alg.~\ref{alg:NetECsingle}) applied to the simulated datasets are depicted in Fig.~\ref{fig:netecepi1000}, where inferred episode sizes are presented as histograms attached to corresponding nodes of network $N$.

Our solution minimizes the \emph{number} of duplication episode locations, not the episode sizes themselves. Because the simulation includes background single-gene duplications and incomplete lineage sorting (ILS), small non-zero episode sizes appear at many nodes even in the absence of any WGD. These counts reflect the spread of background duplications and should not be interpreted as evidence of WGD events. Only episodes with substantially elevated counts provide evidence for WGD.

For WGD scenarios $A$, $C$, and $AC$, NetEC tool successfully inferred the events in all experimental settings where they were simulated. These scenarios exhibited strong signals exceeding 1,000 duplications, whereas in evolutionary scenarios lacking the corresponding WGD, the signal matched the background episode sizes of the null scenario ($\emptyset$). WGD events at locations $A$ and $C$ were also correctly identified with significant episode sizes in scenarios $AB$ and $BC$, respectively.

In contrast, the event at location $B$ was not detected in scenarios $B$, $BC$, and $AB$. We observed that duplications from location $B$ were likely reassigned upward to the LSA node (marked by blue star in Fig.~\ref{fig:netecepi1000}) and partially to the LSA's left~child.

We then performed extended episode inference by examining additional candidates among non-episode nodes (Alg.~\ref{alg:extepisodes}). The results are summarized in the bottom-right panel of Fig.~\ref{fig:netecepi1000}, where the method identified node $B$ as the only extended episode with high support (above 2,000 duplications) in WGD scenarios $B$, $BC$, and $AB$. For the remaining scenarios, episode sizes were small, indicating that the inference was complete.

In the subsequent analysis, we employed NetEC tool with the extended episode at location $B$ for the three WGD scenarios $B$, $BC$, and $AB$. The results are shown in the right side of Fig.~\ref{fig:netecepi1000}, where the number of duplications at location $B$ correctly indicates the simulated WGD event. We also observed that the number of duplications at the LSA node and its child were significantly reduced compared to the previous analysis.

Finally, performing extended episode analysis on this inference did not identify any additional episode candidates (see bottom-right panel of Fig.~\ref{fig:netecepi1000}).

\begin{figure*}
\centering
\includegraphics[width=.99\textwidth]{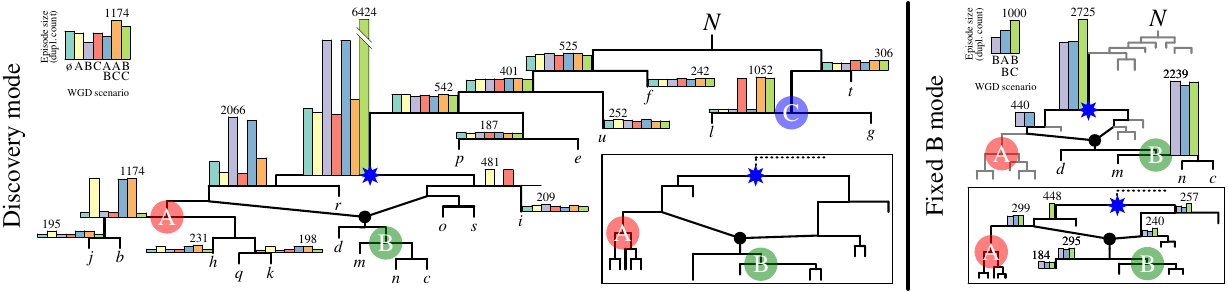}
\caption{\textbf{Inference of duplication episodes from seven simulated WGD scenarios on phylogenetic network $N$ (one reticulation node, black circle).} Nodes $A$, $B$, $C$ indicate simulated WGD positions across seven scenarios: $\emptyset$, $A$, $B$, $C$, $AB$, $AC$, $BC$. Histograms show inferred episode sizes (number of gene duplications); values above indicate maximum frequency; histograms with all values below 150 are omitted. \textbf{Left:} Discovery mode results. \textbf{Middle frame:} Extended analysis (Alg.~\ref{alg:extepisodes}) reveals node $B$ as an additional high-support episode in scenarios $B$, $AB$, and $BC$. \textbf{Right:} Results after fixing node $B$ as an extended episode in scenarios $B$, $AB$, and $BC$. \textbf{Right frame:} Extended analysis confirms no further episode candidates. Note that branch lengths of $N$ are adjusted for clarity across all visualisations and do not reflect its ultrametric structure.}
\label{fig:netecepi1000}
\label{fig:netecepi1000fixb}
\end{figure*}

\begin{figure}
\centering
\includegraphics[width=.99\columnwidth]{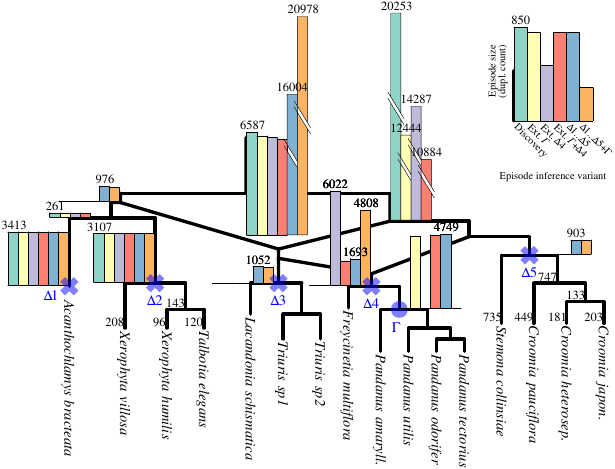}
\caption{\textbf{Episode inference on empirical Pandanales gene trees on a phylogenetic network with two reticulations.}
  The network follows~\cite{shi2025resolving} with five hypothesized WGD events ($\Delta$1--$\Delta$5), blue
  crosses) and an additional candidate node $\Gamma$ (blue circle).
  Bar charts at nodes display episode sizes under six inference variants: discovery mode
   alone, extended analysis forcing $\Gamma$, extended analysis forcing $\Delta$4, combined extended analysis
  ($\Gamma$+$\Delta$4), hypothesis-driven mode with all five published WGDs ($\Delta$1--$\Delta$5), and
  hypothesis-driven mode including $\Gamma$ ($\Delta$1--$\Delta$5+$\Gamma$).
  Numbers above bars indicate episode sizes, while the numbers near nodes without histograms represent equal duplication count obtained in all six inferences.
  Results show strong support for $\Delta$1, $\Delta$2, and either $\Delta$4 or $\Gamma$ (which absorbs signal from
  $\Delta$4 when included), with moderate support for $\Delta$5.
  Deep ancestral nodes accumulate duplications due to upward clustering permitted by reticulation-induced scenario
  flexibility.}
\label{fig:disco}
\end{figure}

\subsection*{Empirical evaluation: Pandanales}

\emph{Data preparation.}
We investigated the placement and support of WGD events in a reticulate evolutionary history of Pandanales. Starting from the Pandanales species tree, we extended it to a phylogenetic network by introducing two reticulation events inferred by Shi et al.~\cite{shi2025resolving} using HyDe-based gene-flow tests. As the duplication signal, we used 29,453 gene trees reconstructed from sequence data available in~\cite{shi2025resolving}; sequences were processed and partitioned into families following the pipeline described in that study, with gene trees inferred using IQ-TREE under the GTR+$\Gamma$ model. We performed two experiments: a discovery run and a hypothesis-driven run in which the five published WGD nodes ($\Delta$1--$\Delta$5, Fig.~\ref{fig:disco}) were provided as user-specified episodes to test whether duplications concentrate at those locations. The analysis by \NetEC tool, including all inference variants, completed within 12~hours.

\emph{Results.}
In discovery mode (Fig.~\ref{fig:disco}; first bars), support concentrates at deep ancestral nodes (right child of root: 20253; root: 6587), reflecting upward clustering below reticulations. Among the published WGDs, $\Delta$1 and $\Delta$2 are directly recovered (3413 and 3107), while $\Delta$3--$\Delta$5 are not. Extended analysis  identified only two high-signal candidates: $\Gamma$ (4599) and $\Delta$4 (6022). Forcing episodes at these nodes individually and jointly (Fig.~\ref{fig:disco}) reveals that $\Delta$4 absorbs more duplications alone, but $\Gamma$ captures the larger local fraction when both are included (4696 vs.\ 1524 at $\Delta$4). This illustrates that extended episodes are best applied iteratively: inserting a high-signal episode and recomputing reveals whether a residual peak warrants an additional episode.

In hypothesis-driven mode, inserting $\Delta$1--$\Delta$5 as user-specified episodes yields strong support for $\Delta$1, $\Delta$2, and $\Delta$4, moderate for $\Delta$5, and weaker for $\Delta$3, with a large fraction remaining at deep nodes. Extended analysis further suggests $\Gamma$ as an additional episode (4749), reducing $\Delta$4 to 1693 and indicating $\Gamma$ as a more appropriate WGD location.

Taken together, minimizing episode count can favor deeper placements when reticulation allows duplications to shift upward, and extended analysis is essential for revealing biologically relevant WGD candidates.

  \section{Discussion and Conclusions}

  Here we propose the first algorithmic framework for identifying duplication episodes in phylogenetic networks. Experiments demonstrated that our approach accurately recovers genomic duplication events in the presence of reticulation, though with important caveats. Events on lineages not directly involved in hybridization are
  consistently detected with strong signal. In contrast, WGD events below reticulation nodes require the extended  episode procedure for reliable detection, as the flexible mapping model allows their duplications to be reassigned
  to ancestral positions. The extended analysis should therefore be considered a standard component of the inference pipeline.

  While the feasibility test runs in polynomial time, the overall algorithm has exponential worst-case complexity; larger networks may require heuristic sampling, yielding upper bounds rather than guaranteed optima. Our model
  assumes correctly rooted gene trees with known leaf-to-taxa mappings; extending it to handle unrooted trees or uncertain taxonomic assignments would broaden applicability to metagenomics. The framework does not explicitly model
   gene losses, which follow WGD events at high rates during diploidization; incorporating losses could improve   detection of ancient events but would likely increase computational complexity.   Incomplete lineage sorting (ILS) is another confounding factor, that can produce gene tree topologies that resemble
  duplications, potentially contributing to background episodes even in the absence of WGD.
  Our simulations include ILS through the multilocus coalescent model, and the background signal observed in the null
  scenario reflects this effect; however, a systematic evaluation of how varying ILS levels affect episode inference, particularly the ability to distinguish true WGD signal from ILS noise, remains a valuable  direction for future work. The model also does not exploit synteny information, i.e.~the chromosomal position and gene order.
  Synteny may provide complementary evidence for distinguishing whole-genome from small-scale duplications and could help refine episode placement; integrating synteny-aware constraints is another avenue for future extension.

  Future directions include establishing the computational complexity of \NetEC problem, which we conjecture to be intractable for general networks, integration with network inference methods for joint topology-episode estimation, statistical frameworks for assessing episode support beyond size metrics, and extension to time-calibrated networks
  where temporal constraints could reduce placement ambiguity. As phylogenomic datasets grow, methods accommodating
  network phylogenies will be essential for understanding complex gene family histories.
  The presented solution represents a step toward this goal, and the publicly available implementation enables both exploratory analysis and targeted evaluation of specific WGD hypotheses.

\section*{Acknowledgements}

Financial support was provided by the National Science Centre grant \#\opusUniNet.

\bibliographystyle{abbrv}
\bibliography{references}

\clearpage

\section*{Appendix: Lemma~\ref{lem:DPcor} - DP correctness}\label{app:lem:DPcor}

\begin{lemma}
\label{lem:DPcor}
Given a gene tree $G$ over a network $N$ and $X \subseteq V(N) \setminus R(N)$.
Let $g \in V(G)$, $s \in V(N)$ and $\dot{s}=s'$ if $s \in R(N)$, and $\dot{s}=s$, otherwise. Then,
\begin{itemize}
\item[P1] $\delta(g,s)=\True$ if and only if there is a feasible mapping $\Mapf$ for $(g,\dot{s},X)$ such that $g$ is an lca-duplication assigned by $\Mapf$ to $\dot{s}$.

\item[P2] $\delta(g,s)=\Unknown$ if and only if there is no feasible mapping for $(g,\dot{s},X)$, but there is a weakly feasible mapping $\Mapf$ for $(g,\dot{s},X)$ such that $g$ is an lca-duplication assigned by $\Mapf$ to $\dot{s}$.

\item[P3] $\sigma(g,s)=\True$ if and only if there is
a feasible mapping $\Mapf$ for $(g,\dot{s},X)$ such that $g$ is a speciation or a leaf  assigned by $\Mapf$ to $\dot{s}$.

\item[P4] For any $g$ and $s$, $\sigma(g,s) \neq \Unknown$.

\item[P5] $\delta^\downarrow(g,s)$ is $\True$ if and only if
there is a feasible mapping for $(g,\dot{s},X)$.

\item[P6] $\delta^\downarrow(g,s)$ is $\Unknown$ if and only if
there is no feasible mapping for $(g,\dot{s},X)$, but there is a weakly feasible mapping for $(g,\dot{s},X)$.
\end{itemize}
\end{lemma}

The auxiliary functions $\epsilon$ and $\delta^*$ are treated as local (no separate arrays in implementation). If $\epsilon(g,s)=\True$, there is a feasible mapping $\Mapf$ for $(g,s,X)$ with $\Mapf^*(g)=s$; if $\epsilon(g,s)=\Unknown$, then $\sigma(g,s)=\False$ and $\delta(g,s)=\Unknown$, and $g$ is still assigned to $s$. $\delta^*(g,s)=\True$ requires a feasible mapping where $g$ is a duplication at $s$ with at least one child also mapped to $s$ (modelled by $\epsilon$); $\delta^*(g,s)$ resembles $\delta(g,s)$ but constrains only non-upper duplications when $\Unknown$.

\newcommand{\mt}{\Mapf}
\newcommand{\mtp}{\Mapf'}
\newcommand{\mtpp}{\Mapf''}
\newcommand{\vn}{V(\hat{N})}

\newcommand{\forg}{(g,s,X)}
\newcommand{\forgp}{(g',s,X)}
\newcommand{\forgpp}{(g'',s,X)}
\newcommand{\Mapfgp}{\Mapf|g'}
\newcommand{\Mapfgpp}{\Mapf|g''}

\subsection*{Preliminary properties}

Before starting the main proof, we need several additional notions and properties.

To preserve validity of mappings when merging mappings from subtrees, we need an operation that shifts a mapping in unfolded network to a proper copy of a subtree in $\hat{N}$. For a valid mapping $F \colon V(G|g) \rightarrow \vn$ such that $F(g)=v$
and $w \in \pi^{-1}(v)$, by $F_w \colon V(G|g) \rightarrow \vn$ we denote
a mapping defined using a top-down approach. Let $F_w(g):=w$.
If $x$ is a child of $y \preceq g$, then the path $p_1,p_2,\dots,p_k$ from $p_1=F(y)$ to $p_k=F(x)$ induces the path $P = \pi(p_1),\dots,\pi(p_k)$ in $N$. If $F_w(y)$ is already determined, then $F_w(x)$ is the unique node in $\hat{N}$
such that the path from $F_w(y)$ to $F_w(x)$ in $\hat{N}$ induces the path equal to $P$ in $N$. Note that $\pi(F)=\pi(F_w)$, however, the main property is stronger.
It follows from the construction that both mappings induce the same paths in a network $N$.

\begin{figure}[!htbp]
\centering\includegraphics[width=0.45\textwidth]{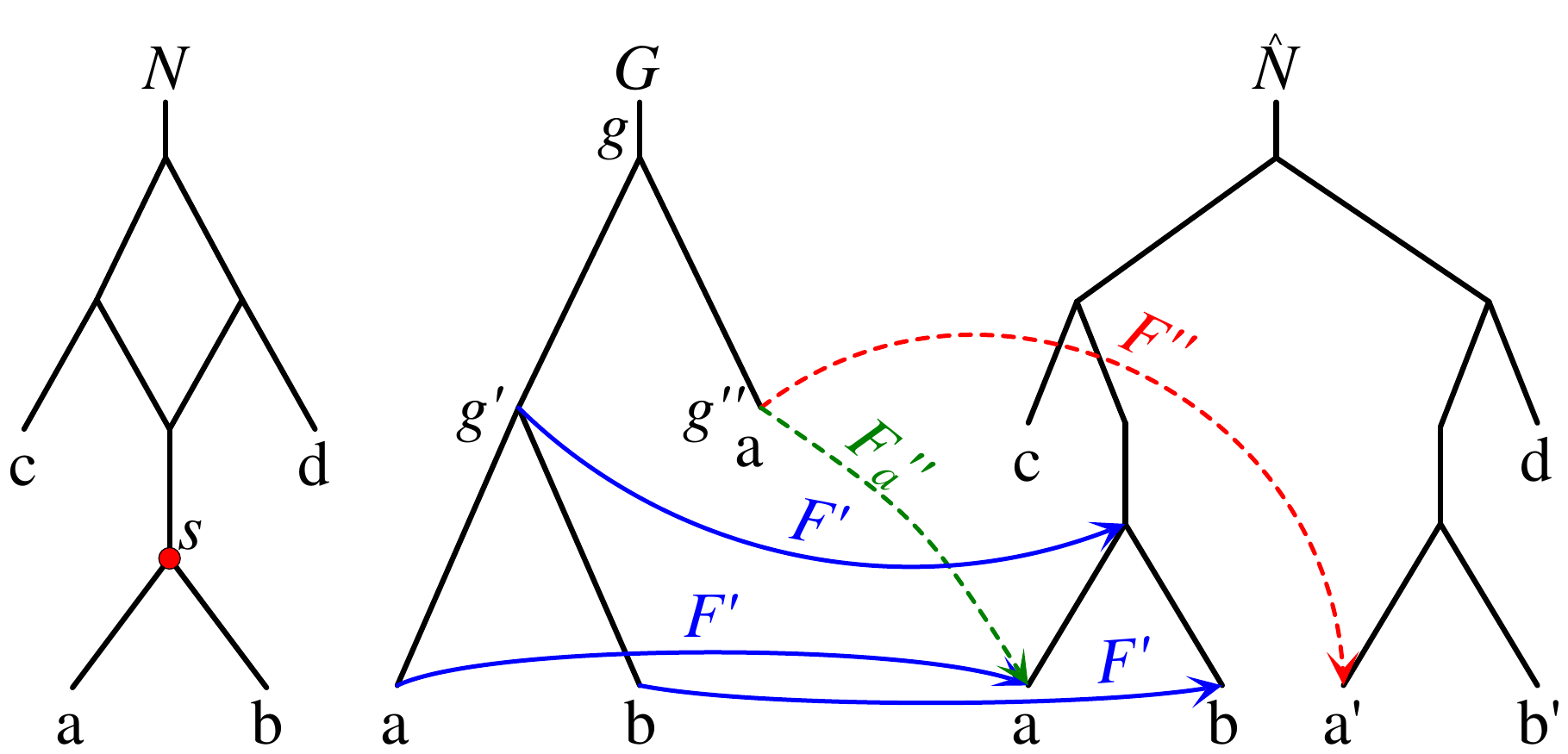}
\caption{Network $N$ and mappings $F'$, $F''$ and $F''_a$.
$F'$ is valid for $(g',s,X)$, while $F''$ and $F''_a$ are valid for $(g'',s,X)$, where $X=\{s\}$. Note that $F' \oplus F''$ is valid for $(g,s,X)$, but $F''$ has to shifted (see  $F''_a$) such that $g''$ maps to $a$ (not $a'$) in $\hat{N}$. Then, in $F' \oplus F''$, $g$ is a duplication assigned to $s$. Without the shift, the resulting product of $F'$ and $F''$ would map $g$ to the root of $\hat{N}$.}
\label{fig:shift}
\end{figure}

Since functions are relations, we identify a function $f \colon A \to B$
with the sets of pairs $\{ (x, f(x)) \colon x \in A \}$.
Let $\mtp \colon V(G|g') \to V(\hat{N})$ and $\mtpp \colon V(G|
g'')\to V(\hat{N})$ be valid mappings for $(g',s,X)$ and $(g'',s,X)$, respectively,  such that $\mtp^*(g') \succeq \mtpp^*(g'')$. Then, by $\mtp \oplus \mtpp$
we denote the mapping $\mtp|_{V(G|g') \setminus \dsetupper(\mtp)} \cup
 \mtpp_w|_{V(G|g'') \setminus \dsetupper(\mtpp)}
\cup \{(g,v)\} \cup (\dsetupper(\mtp) \cup \dsetupper(\mtpp)) \times \{v\}$ where $\mtp(g)=v$ and $\mtpp_w \colon V(G|g'') \arrow V(\hat{N})$ is the copy of $\mtpp$ shifted to be present below $v$ in $\hat{N}$, where $w$ is a node such that $\pi(w)=\mtpp^*(g'')$ and $v \succeq w$. Such a node exists since $\mtp^*(g') \succeq \mtpp^*(g'')$.
\begin{lemma}
\label{lem:otimes}
Let $\mtp \colon V(G|g') \to V(\hat{N})$ and $\mtpp \colon V(G|
g'')\to V(\hat{N})$ be valid mappings for $(g',s,X)$ and $(g'',s,X)$, respectively, such that $\mtp^*(g') \succeq \mtpp^*(g'')$. Then, $\mtp \oplus \mtpp$ is valid for a  $(g,s,X)$ with $g$ being a duplication assigned to~$s$.
\end{lemma}
\begin{proof}
It follows from the definition of the mapping.
\end{proof}
Note that, without the shift operation, $g$ may be assigned to a node above $s$. See example in Fig.~\ref{fig:shift}.

Similarly, to $\mtp \oplus \mtpp$, we define a mapping $\mtp \otimes \mtpp \colon V(G|g) \rightarrow \vn$.
Here, $\mtp$ is a mapping feasible for $(g',s',X)$, $\mtpp$ is a mapping feasible for $(g'',s'',X)$, and $s'$ and $s''$ are the children of $s$ in $N$. In such a case the mapping is defined as $\mtp \cup \mtpp_{w} \cup \{(g,v)\}$,
where $v$ is the parent of $F(g')$ in $\hat{N}$ (note that $\pi(v)=s$)
and $w$ is defined as in the previous paragraph. The proof of the next Lemma follows immediately from the definition of the mapping.
\begin{lemma}
\label{lem:otimes2}
Let $\mtp$ and $\mtpp$ be mappings feasible for $(g',s',X)$ and $(g'',s'',X)$,  respectively, where $s'$ and $s''$ are the children of $s$ in $N$.
Then, $\mtp \otimes \mtpp$ is valid for a gene tree $G|g$, where $g$ is a speciation assigned to $s$.
\end{lemma}

\subsection*{The main proof}

The proof of Lemma~\ref{lem:DPcor} is by induction on the structure of $G$ and $N$.
The base of induction is when $g \in L(G)$ and $s \in L(N)$, for which all properties are easy to verify.

\emph{Inductive assumption:} For every $x$, $y$ such that $g \succ x$,
and $s \succeq y$ or $g \succeq x$ and $s \succ y$, P1-P6 are satisfied.
\emph{Inductive hypothesis:} For $g$ and $s$, where at least one of $g$ and $s$ is not a leaf, P1-P6 are satisfied.

First, if $s \in R(N)$ then the properties follow immediately from the inductive assumption for $g$ and $s'$, and identities~\eqref{vdretcase}.
Therefore, in the next part, we assume that $s$ is not a reticulation, and we prove the properties using $s$ instead of $\dot{s}$.

We start with several properties.

(A1) if $\delta^\downarrow(g',s)=\Unknown$, then $s \notin X$. Assume that $s \in X$.
Then, by P6, there is a weakly feasible mapping $\mtp \colon V(G|g') \arrow V(\hat{N})$ for $(g',s, X)$. Since, $\mtp^*(g) \preceq s$, there must be $v \in
\vn$ such that $\mtp(g) \preceq v$ and $\pi(v)=s$.
Now, let $\mt := \mtp|_{V(G|g') \setminus \dsetupper(\mtp)} \cup \dsetupper(\mtp) \times \{v\}$. It is not difficult to see that $\mt$ is feasible for $(g',s, X)$, since all duplications from $\dsetupper(\mtp)$ are assigned to $s$. A contradiction.

(A2) If $\epsilon(g',s)=\Unknown$, then there is a weakly feasible  $\mtp$ for $(g',s,X)$ and $g'$ is an lca-duplication assigned to $s$. Here, $\sigma(g',s)$ cannot be $\True$, thus $\sigma(g',s)=\False$, by P3 and P4. Therefore, $\delta(g',s)=\Unknown$. The rest follows from P2.

We first prove properties P1-P4 for $\delta$ and $\sigma$. Then, we show that P5 and P6 hold for $g$ and $s$.

(P1, $\Rightarrow$): If $\delta(g,s)=\True$, then, from (\ref{v1d})
$g$ is internal, $s \in X$ and $\delta^*(g,s) = \True$.
Then, w.l.o.g., for a child $g'$ of $g$,
$\epsilon(g',s) \wedge \delta^\downarrow(g'',s)=\True$.
Since, $\delta(g',s) \vee \sigma(g',s)$ is $\True$, it follows from the inductive assumption for P1 and P3, that there is a feasible mapping
$\mtp$ for $(g',s,X)$ such that $\mtp^*(g')=s$ and $g'$ is either speciation or lca-duplication. For the other child,
we have $\delta^\downarrow(g'',s)=\True$. From P5, there is a feasible mapping $\Mapf''$ for $(g'',s,X)$. Then, the mapping $\mtp \oplus \mtpp$ is feasible for $(g,s,X)$ where $g$ is an lca-duplication assigned to $s$.

(P1, $\Leftarrow$):
Assume there is a feasible mapping
$\Mapf \colon V(G|g) \arrow \vn$ for $(g,s,X)$ such that $g$ is~an~lca-duplication assigned to $s$. Thus, $s \in X$ and $g$ is internal. We conclude that the condition from  (\ref{v1d}) is satisfied, and  $\delta(g,s)=\delta^*(g,s)$. W.l.o.g. we assume
that $g'$ is assigned to $s$ (recall that $g$ is an lca-duplication). Then,
$\Mapfgp$ is feasible for $(g',s,X)$.
By the inductive assumption for P1 (if $g'$ is a duplication) or P3 (if $g'$ is a speciation or a leaf), we conclude that $\epsilon(g',s)=\True$. For the second child, we have $F''^*(g'') \preceq s$, thus
$\Mapfgpp$ is feasible for $(g'',s,X)$ and by P5, $\delta^\downarrow(g'',s)=\True$. Finally, $\delta^*(g,s)=\True = \delta(g,s)$.

(P2, $\Rightarrow$): Let $\delta(g,s)=\Unknown$.

(Case P2.a) If $s \in X$ then, from (\ref{v1d}) $g$ is internal and $\delta^*(g,s) = \Unknown$.  W.l.o.g., we may assume that $\epsilon(g',s) \wedge \delta^\downarrow(g'',s)=\Unknown$,
thus $\delta^\downarrow(g'',s) \succeq \Unknown$. If $\delta^\downarrow(g'',s)=\Unknown$
then $s \notin X$ from (A1), we conclude that $\delta^\downarrow(g'',s)=\True$ and $\epsilon(g',s)=\Unknown$.
By (A2) there is a weakly feasible mapping $\Mapf'$ for $(g',s,X)$ and $g'$ is an lca-duplication assigned to $s$. Since $s \in X$, we can construct a feasible mapping for $(g',s,X)$ from a weakly feasible $\mtp$, by assigning all upper duplications from $\mtp$ to $s$. A contradiction.

(Case P2.b) Assume that $s \notin X$. Then, from (\ref{v2d}) $g$ is internal and
$\delta^*(g,s) \wedge \Unknown = \Unknown$. We have that $\delta^*(g,s) \in \{ \True, \Unknown \}$. (Case P2.b.1) Let $\delta^*(g,s)=\True$.
W.l.o.g. assume that $\epsilon(g',s)=\delta^\downarrow(g'',s) = \True$.
Then, $\epsilon(g',s) = \sigma(g',s) \vee \delta(g',s)=\True$.
From P1 and $s \notin X$, $\delta(g',s)$ cannot be $\True$, thus $\sigma(g',s)=\True$. From P3, there is a feasible
$\mtp$ for $(g',s,X)$ and $g'$ is an speciation or a leaf assigned to $s$.
For the second child we have $\delta^\downarrow(g'',s) = \True$ and from P5, there is a feasible $\mtpp$ for $(g'',s,X)$.
Let $\Mapf = \mtp \oplus \mtpp$. It is not difficult to see that $g$ is an lca-duplication in $\Mapf$ since $g'$ is speciation or a leaf. Thus, $\Mapf$ is weakly feasible. Note that there is no feasible mapping for $(g,s,X)$; otherwise, $x \in X$ based on the property that $g$ is an lca-duplication assigned to $s$.

(Case P2.b.2) Let $\delta^*(g,s)=\Unknown$.
W.l.o.g. assume that
$\epsilon(g',s) \wedge \delta^\downarrow(g'',s) = \Unknown$.
If $\epsilon(g',s)=\True$ then $\delta^\downarrow(g'',s)=\Unknown$. Then, similarly to the previous case $g'$ is a speciation or a leaf assigned to $s$, and there is a feasible $\mtp$ for $(g',s,X)$, while, from P6, there is a weakly feasible $\mtpp$ for $(g'',s,X)$. Then, similarly to the previous case
$\mtp \oplus \mtpp$ is weakly feasible mapping for $(g,s,X)$ where $g$ is an lca-duplication assigned to $s$.
It remains to analyse the case when $\epsilon(g',s)=\Unknown$.
By (A2) there is a weakly feasible  $\mtp$ for $(g',s,X)$ and $g'$ is an lca-duplication assigned to $s$. Here, $\delta^\downarrow(g'',s) \in \{\True,\Unknown\}$, and depending on the value either, by P5 there is a feasible ($\True$) or, by P6, weakly feasible ($\Unknown$) $\mtpp$ for $(g'',s,X)$.
In the valid mapping $\mtp \oplus \mtpp$, $g'$ is an lca-duplication assigned to $s$,
and $g$ is also lca-duplication assigned to $s$. However, $s \notin X$, thus the mapping
is weakly feasible for $(g,s,X)$.
Also, there is no feasible mapping for $(g,s,X)$ in this case, due to the property
that $g'$ is an lca-duplication assigned to $s$.

(P2, $\Leftarrow$). Let $\mt$ be a weakly feasible mapping for $(g,s,X)$ such that $g$ is an lca-duplication assigned to $s$.
Since $g$ is an lca-duplication, $g$ is also an upper duplication and $s \notin X$.
Thus, $\delta(g,s)=\delta^*(g,s) \wedge \Unknown$ from $(\ref{v2d})$. W.l.o.g., we may assume that $g'$ is assigned to $s$. If $g'$ is a speciation or a leaf,
then $\Mapfgp$ is feasible for $\forgp$ since no upper duplication is present in $G|g'$. From P3, $\sigma(g',s)=\epsilon(g',s)=\True$.
If $g'$ is a duplication, $g'$ is an upper lca-duplication, thus $\Mapfgp$ is weakly feasible. From P2, $\delta(g',s)=\Unknown$.
In all cases, $\epsilon(g',s) \geq \Unknown$.
Similarly, for the second child of $g$, the mapping
$\Mapfgpp$ is either weakly feasible for $\forgpp$ if $g''$ is an upper duplication
assigned to a node not in $X$, or feasible otherwise.
By P5 and P6, $\delta^\downarrow(g'',s) \geq \Unknown$.
Finally, $\delta^*(g,s) \geq \Unknown$ and $\delta(g,s)=\delta^*(g,s) \wedge \Unknown = \Unknown$.

(P3, $\Rightarrow$). Let $\sigma(g,s)=\True$.
Note that at least one of $g$ and $s$ is internal by the inductive assumption. Then, $g$ is internal and $s$ is a tree-node, from (\ref{v1s}).  W.l.o.g. we may assume that $\delta^\downarrow(g',s') \wedge \delta^\downarrow(g'',s'')=\True$.
Thus, from P5, we have two feasible mappings $\Mapf'$ for $(g',s',X)$
and $\Mapf''$ for $(g'',s'',X)$. Then, the mapping
$\Mapf' \otimes \Mapf''$ is feasible for $(g,s,X)$
where $g$ is  a speciation assigned to $s$.

(P3, $\Leftarrow$).
Assume that there is a feasible mapping $\Mapf$ for $(g,s,X)$ such that $g$ is a speciation or a leaf assigned to $s$. If $g$ is a leaf, the statement is obvious.
Assume that $g$ is a speciation, then $s$ is a tree-node and
$\sigma(g,s)$ follows from (\ref{v1s}).
W.l.o.g. we may assume that $\Mapf^*(g') \preceq s'$ and
$\Mapf^*(g'') \preceq s''$.
Thus, $\Mapfgp$ is feasible for $(g',s',X)$.
From P5, $\delta^\downarrow(g',s')=\True$. Similarly, we obtain $\delta^\downarrow(g'',s'')=\True$.
Finally, $\sigma(g,s) = \Lop ( \delta^\downarrow(g',s') \wedge  \delta^\downarrow(g'',s'')) = \True$.

(P4) It follows easily from the definition of $\sigma$ and the operator $\Lop$.

(P5, $\Rightarrow$) Assume that $\delta^\downarrow(g,s)=\True$.

(Case P5.1) If $s$ is a leaf, then $g$ is internal by the inductive assumption.
Then, by (\ref{v1dd}), $\delta^\downarrow(g,s)=\epsilon(g,s)=\delta(g,s) \vee \sigma(g,s) = \True$. Note, that $\sigma(g,s)=\False$, otherwise both $g$ and $s$
are leaves. Thus, $\delta(g,s)=\True$ and the feasible mapping for $(g,s,X)$ exists by the already proven P1.

(Case P5.2). If $s$ is a tree-node and $s \notin X$, then, by (\ref{v1dd}),
one of $\delta(g,s)$, $\sigma(g,s)$, $\delta^\downarrow(g,s')$ or $\delta^\downarrow(g,s'')$ is $\True$.
If $\delta(g,s)=\True$, there is a feasible mapping for $(g,s,X)$ from already proven P1 for $g$ and $s$.
Similarly, we have the mapping from P3 if $\sigma(g,s)=\True$.
If $\delta^\downarrow(g,s')=\True$, then there is a feasible mapping $\Mapf$ for $(g,s',X)$ from the inductive assumption for P5. Clearly, $\Mapf$ is also feasible
for $(g,s,X)$.
The remaining case when $\delta^\downarrow(g,s'')=\True$ is analogous.

(Case P5.3). If $s$ is a tree-node and $s \in X$, then, by (\ref{v2dd}),
at least one among $\delta(g,s)$, $\sigma(g,s)$, $\Mop \delta^\downarrow(g,s')$, and $\Mop \delta^\downarrow(g,s'')$ is $\True$.
The proof is the same as above when $\delta(g,s)=\True$,
$\sigma(g,s)=\True$, $\delta^\downarrow(g,s')=\True$ or $\delta^\downarrow(g,s'')=\True$.
For the remaining case, assume that $\delta^\downarrow(g,s')=\Unknown$. Then there is a weakly feasible mapping $\Mapf$ for $(g,s',X)$ from the inductive assumption for P6. By reassigning all upper $\Mapf$-duplications to $s$ we construct a feasible mapping for $(g,s,X)$. The remaining case when $\delta^\downarrow(g,s'')=\Unknown$ is analogous.

(P5, $\Leftarrow$)
Assume there is a feasible mapping $\Mapf$ for $(g,s,X)$.

(Case P5.1) If $s$ is a leaf, then $g$ is internal by the inductive assumption. Thus, $g$ is an upper lca-duplication assigned to $s$ under the mapping.
Thus, $s \in X$ and by already proven P1, $\delta(g,s)=\True$. This yields $\epsilon(g,s)=\True=\delta^\downarrow(g,s)$ in this case by~\eqref{v1dd}.

(Case P5.2) Let $s$ be a tree-node. If $g$ is an lca-duplication assigned to $s$
then similarly to the above case, from P1, $s \in X$, and $\delta(g,s)=\True$ and $\delta^\downarrow(g,s)=\True$ using (\ref{v2dd}).
The proof is analogous when $g$ is a speciation assigned to $s$. The only difference is that $s$ need not be in $X$.

For the remaining cases, we have either $g$ is a duplication assigned to $s$, but not the lca-duplication, or there exists $v \prec s$, such that $g$ is either a duplication, or a speciation assigned to $v$.

(Case P5.2.a) Let $g$ be a speciation assigned to a node $v$ in $N$.
W.l.o.g., we may assume that $v \preceq s' \prec s$.
Then $\Mapf^*(g)=v$ and there is no upper duplication in $T$. Clearly, $\Mapf$ is also feasible for $(g, s',X)$. From, the inductive assumption for P5, $\delta^\downarrow(g,s')=\True$ and also $\Mop \delta^\downarrow(g,s') = \True$. In all cases, from \eqref{v2dd} if $s \in X$,
and \eqref{v1dd}, otherwise, we obtain $\delta^\downarrow(g,s)=\True$.

(Case P5.2.b) Assume that $g$ is a duplication assigned to $s$, but not the lca-duplication. Thus, $s \in X$ and $\Mapf^*(g)=s$. Let $w:=\Map_\xi(g) \in \hat{N}$, where $\xi=\Mapf|_{L(G)}$. In other words, $w$ is the lca-mapping of $g$ in $\hat{N}$.
Since, $g$ is not lca-duplication, we have that $\pi(w) \prec s$. Let $v'$ be a child
of $\Mapf(g)$ present on the path from $\Mapf(g)$ to $v$ in $\hat{N}$. Then, $\pi(v')$ is a child, say $s'$, of $s$ in $N$.
We construct a mapping $\Mapf'$ such that
$\Mapf'(u):=v'$, if $u$ is upper duplication induced by $\Mapf$, and $\Mapf'(u):=\Mapf(u)$, otherwise.
If $s' \in X$, then such a mapping is feasible for $(g,s',X)$ since all upper duplications are assigned to $s'$, otherwise, the mapping is weakly feasible for $(g,s',X)$. Note that in the second case, there may exist alternative feasible mappings for $(g,s,X')$. In both cases, we get from the inductive assumption for $P5$ and $P6$ that $\delta^\downarrow(g,s') \geq \Unknown$, which gives $\Mop \delta^\downarrow(g,s')=\True$ (note that $s\in X$). Thus, $\delta(g,s)=\True$.

(Case P5.2.b) Assume that $g$ is a duplication assigned to a node strictly below $s$.
Similarly to the previous case, we define $w$ and identify $s'$ as the child of $s$
such that there is a path from $s'$ to $\pi(w)$.
We conclude that $\Mapf^*(g) \preceq s'$. Thus, $\Mapf$ is a feasible mapping for
$(g,s',X)$. The rest follows in the same way.

(P6, $\Rightarrow$) Assume, $\delta^\downarrow(g,s)$ is $\Unknown$.
Here, $\Unknown$ is obtained only from (\ref{v1dd}), i.e., when $s \notin X$. In addition, $s$ is a tree-node, i.e., $ch(s)$ is not empty.
After expanding $\epsilon$, $\delta(g,s) \vee \sigma(g,s) \vee \delta^\downarrow(g,s') \vee \delta^\downarrow(g,s'')$ is $\Unknown$.
By P4, $\sigma(g,s)=\False$.
If $\delta(g,s)=\Unknown$, then there is no feasible mapping for $(g,s,X)$ but there is a weakly feasible mapping from already proven P2.
If $\delta^\downarrow(g,s')=\Unknown$, then, by the inductive assumption for P6, there is no feasible mapping for $(g,s',X)$, but there is a weakly feasible mapping
for $(g,s',X)$. Since, $s \notin X$, these properties are also satisfied for $(g,s,X)$.
A similar argument hold when $\delta^\downarrow(g,s'')=\Unknown$.

(P6, $\Leftarrow$)
Assume there is no feasible mapping for $(g,s,X)$, but there is a weakly feasible mapping $\Mapf$ for $(g,s,X)$. We show that $\delta^\downarrow(g,s)=\Unknown$.

(Case P6.1) If $s$ is a leaf, then $g$ is internal by the inductive assumption. Thus, $g$ is an upper duplication assigned to $s$.
Thus, $s \notin X$ and by already proven P2, $\delta(g,s)=\Unknown$. Since, $\sigma(g,s)=\False$, $\epsilon(g,s)=\Unknown$ and $\delta^\downarrow(g,s)=\Unknown$ from (\ref{v1dd}).

(Case P6.2) Assume $s$ is internal. Note that $g$ cannot be a speciation or a leaf, otherwise there is no upper duplication in $T$ and $\Mapf$ is not weakly feasible.
If $g$ is an lca-duplication assigned to $s$,
then similarly to the previous case, from P2, $s \notin X$, and $\delta(g,s)=\Unknown$. Note that $\delta^\downarrow(g,s')=\delta^\downarrow(g,s'')=\False$
(since $g$ is lca-duplication). We conclude that
$\delta^\downarrow(g,s)=\Unknown$ from (\ref{v1dd}).

For the remaining case, $g$ is a duplication, and
either $g$ is assigned to $s$ and $g$ it is not lca-duplication or
$g$ is assigned to a node strictly below $s$.
Also $s \notin X$, otherwise there is a feasible mapping for $(g,s,X)$
obtained from $\Mapf$ by assigning all upper duplications to $s$.
The mapping $\Mapf$ is weakly feasible for $(g,s',X)$, where $s'$ is the node in $N$
from which $\pi(\Mapf(g))$ is visible (see P5.2, $\Leftarrow$ for the construction of $s'$). Clearly, there is no feasible mapping for $(g,s',X)$,
otherwise such a mapping would be feasible for $(g,s,X)$.
By the inductive assumption for P6, we conclude that
$\delta^\downarrow(g,s')=\Unknown$. For the second child of $s$,
we also have $\delta^\downarrow(g,s'') \neq \True$,
otherwise, by P5, $\Mapf$ is feasible for $(g,s'',X)$ and also for $(g,s,X)$, a contradiction. Similarly, we have that $\delta(g,s) \neq \True$
and $\sigma(g,s)=\False$. Thus, $\epsilon(g,s) \leq \Unknown$.
Finally, $\delta^\downarrow(g,s)=\Unknown$, which follows from \eqref{v1dd}.
This completes the proof of all properties P1-P6 and Lemma~\ref{lem:DPcor}.

\section*{Appendix: the proof of Thm~\ref{lem:DPcormain}}
\begin{proof}
The proof follows immediately from P5 of Lemma~\ref{lem:DPcor}:
$\delta^\downarrow(\troot(G),\troot(N))$ is $\True$ if and only if
there is a feasible mapping $\Mapf$ for $(\troot(G),\troot(N),X)$.
In such a case $\epi(\Mapf) \subseteq X$.
\end{proof}

\section*{Appendix: the proof of Lemma~\ref{lem:genNetEC}}
\begin{proof}
If there is a feasible solution in the outgrouped case, then there is an episode at the root of network $N^\omega$. Removing the inserted nodes yields a feasible solution to the original problem instance. The other direction follows similarly.
\end{proof}

\section*{Appendix: Outgroup construction for multiple gene trees (example)}

We illustrate the outgroup construction from Lemma~\ref{lem:genNetEC} on a small instance with two gene trees $G_1 = (a,(b,c))$ and $G_2 = ((a,b),c)$ over taxa $\{a,b,c\}$ and a network $N$.

Introduce an outgroup species $\omega \notin \{a,b,c\}$. Extend the network to $N^\omega$ by adding a new root $r^\omega$ with children $\omega$ (leaf) and the original root of $N$. For each gene tree, append $\omega$ as a sibling of its root: $G_1^\omega = (G_1,\omega)$ and $G_2^\omega = (G_2,\omega)$. Then merge iteratively: $G_2^\omega = (G_2^\omega, G_1^\omega)$, so the root of the merged tree has $G_2^\omega$ and $G_1^\omega$ as children. This root is mapped to $r^\omega$ as a speciation, and the two subtrees independently resolve $G_1$ and $G_2$ inside $N$.

Applying Alg.~\ref{alg:NetECsingle} to $(G^\omega_2, N^\omega)$ and removing the forced episode at $r^\omega$ yields the optimal episodes for the original instance $(G_1, G_2, N)$, as stated in Lemma~\ref{lem:genNetEC}.

\section*{Appendix: Lemma~\ref{lemma:retexlude} (Reticulation Exclusion)}

This lemma justifies a key optimization in the DP: since a reticulation node has a unique child in $\hat{N}$, a duplication mapped to the reticulation is indistinguishable from one mapped to its child, so we may always assume episode locations lie outside $R(N)$.

\begin{lemma}
\label{lemma:retexlude}
For a valid mapping $\Mapf$ between $G$ and $\hat{N}$, let $\Mapf'$
be the mapping such that for $g \in V(G)$, $\Mapf'(g)$
is the child of $\Mapf(g)$ if $\pi(\Mapf(g)) \in R(N)$,
and $\Mapf'(g):=\Mapf(g)$, otherwise.
Then, $\Mapf'$ is valid, $\epi(\Mapf') \cap R(N) = \emptyset$ and $|\epi(\Mapf)| \geq |\epi(\Mapf')|$.
\end{lemma}
\begin{proof}
We show that $\Mapf'$ is valid, which follows from
the fact that there is no speciation assigned to a reticulation node, since it has a single child. Thus, any node $g$ assigned to a reticulation $s$ via $\Mapf$ must be a duplication. Moreover, every child $g'$ of $g$ is either a duplication assigned to $s$
or below, or it is a speciation assigned to the only child $s'$ of $s$ or below. Thus, mappings of all duplications assigned to $s$ can be lowered to be assigned to $s'$. Now, all four conditions of mapping validity for $\Mapf'$ follow easily.
\end{proof}
\end{document}